\documentclass[sigplan]{acmart}

\usepackage[utf8]{inputenc}
\acmConference[DRAFT]{}{ver.~1.0}{2020}
\acmYear{2020}
\acmISBN{} 
\acmDOI{} 

\setcopyright{none}

\usepackage{booktabs}   
\usepackage{subcaption} 

\usepackage{amsmath}
\usepackage{amssymb}
\usepackage{mathtools}
\usepackage{url}
\usepackage[all]{xy}
\usepackage{latexsym}
\usepackage{xspace}

\usepackage{graphicx}

\usepackage{xspace}

\input diagxy
\xyoption{color}

\usepackage{algorithm}

\usepackage[noend]{algpseudocode}

\newtheorem{theorem}{Theorem}[section]

\newtheorem{definition}{Definition}[section]
\newtheorem{example}{Example}[section]

\newtheorem{remark}{Remark}[section]
\newtheorem{axiom}{Axiom}

\usepackage{tikz}
\usetikzlibrary{automata,positioning}

\usetikzlibrary{calc}
\usetikzlibrary{arrows, decorations.markings}
\usetikzlibrary{patterns}
\tikzstyle{vecArrow} = [thick, decoration={markings,mark=at position
   1 with {\arrow[semithick]{open triangle 60}}},
   double distance=1.4pt, shorten >= 5.5pt,
   preaction = {decorate},
   postaction = {draw,line width=1.4pt, white,shorten >= 4.5pt}]

\tikzstyle{edgeArrow} = [semithick, black,line width=1.4pt]
\tikzstyle{edgeArrowR} = [semithick, black,line width=3pt]

\tikzstyle{edgeArrowD} = [thick, decoration={markings,mark=at position
   1 with {\arrow[semithick]{open triangle 60}}},
   black,line width=1.4pt]

\tikzstyle{vecArrow2} = [thick, black, double distance=1.4pt, shorten >= 1.5pt,
preaction = {decorate},
   postaction = {draw,line width=1.4pt, white,shorten >= 1.5pt}]
   
\tikzstyle{innerWhite} = [semithick, white,line width=1.4pt, shorten >= 4.5pt]

\makeatletter
\def\slashedarrowfill@#1#2#3#4#5{%
  $\m@th\thickmuskip0mu\medmuskip\thickmuskip\thinmuskip\thickmuskip
  \relax#5#1\mkern-7mu%
  \cleaders\hbox{$#5\mkern-2mu#2\mkern-2mu$}\hfill
  \mathclap{#3}\mathclap{#2}%
  \cleaders\hbox{$#5\mkern-2mu#2\mkern-2mu$}\hfill
  \mkern-7mu#4$%
}
\def\rightslashedarrowfill@{%
  \slashedarrowfill@\relbar\relbar\mapstochar\rightarrow}
\newcommand\xslashedrightarrow[2][]{%
  \ext@arrow 0055{\rightslashedarrowfill@}{#1}{#2}}
\makeatother

\newcommand{\group}[1]  {
  \mathbb{#1}
}

\newcommand{\catl}[1]  {
  \mathbb{#1}
}
\newcommand{\catw}[1]  {
  \mathbf{#1}
}
\newcommand{\word}[1]  {
  \mathit{#1}
}
\newcommand{\mor}[3]  {
  #1 \colon #2 \rightarrow #3
}
\newcommand{\tuple}[1]  {
	\langle #1 \rangle
}
\newcommand{\id}[1]  {
  \word{id}_{#1}
}
\newcommand{\cont}[1]  {
  \catw{Cont}(\group{#1})
}
\newcommand{\classifying}[1]  {
  \catw{Set}[T]
}
\newcommand{\aut}[1]  {
  \mathit{Aut}(#1)
}

\newcommand{\struct}[1]  {
  \mathcal{#1}
}

\begin{document}

\title[On amenability of constraint satisfaction problems]{On amenability of constraint satisfaction problems}         


\author{Michal R.~Przybylek}
\affiliation{
  \department{Faculty of Mathematics, Informatics and Mechanics}              
  \institution{University of Warsaw}            
  \city{Warsaw}
  \country{Poland}                    
}
\email{mrp@mimuw.edu.pl}          

\begin{abstract}
The authors of \cite{DBLP:conf/lics/KlinKOT15} showed that a constraint satisfaction problem (CSP) defined over rational numbers with their natural ordering has a solution if and only if it has a definable solution. Their proof uses advanced results from topology and modern model theory. The aim of this paper is threefold. (1) We give a simple purely-logical proof of their theorem and show that the advanced results from topology and model theory are not needed; (2) we introduce an intrinsic characterisation of the statement ``definable CSP has a solution iff it has a definable solution'' and investigate it in general intuitionistic set theories (3) we show that the results from modern model theory are indeed needed, but for the implication reversed: we prove that ``definable CSP has a solution iff it has a definable solution'' holds over a countable structure if and only if the automorphism group of the structure is extremely amenable.
\end{abstract}



\maketitle

\section{Introduction}
\label{s:introduction}

Nowadays, there is no longer any question that computer-aided solutions to real-world problems are critical for the industry. Even relatively small problems can have high complexity, what makes them intractable for human beings. Very many real-world decision problems of high complexity can be abstractly specified as constraint satisfaction problems e.g.~hardware verification and diagnosis: \cite{clarke2003sat}, \cite{gotlieb2012tcas}, automated planning and scheduling \cite{do2001planning}, \cite{Fox-1990-13142}, temporal and spatial reasoning \cite{renz2007qualitative}, \cite{bodirsky2007qualitative}, air traffic managment \cite{allignol2012constraint},... to name a few. A constraint satisfaction problem (CSP) can be abstractly defined as a triple $\tuple{D, V, C}$, where:
\begin{itemize}
\item $D$ is the domain
\item $V$ is the set of variables 
\item $C$ is a set of constraints of the form $\tuple{\tuple{x_1, x_2, \cdots, x_k}, R}$, where $x_i \in V$ and $R \subseteq D^k$
\end{itemize}
A solution to this problem is an assignment $\mor{S}{V}{D}$ that satisfies all constraints in $C$, i.e.: for every $\tuple{\tuple{x_1, x_2, \cdots, x_k}, R} \in C$ we have that $R(S(x_1), S(x_2), \cdots, S(x_k))$ holds. Classical and best explored variant of CSP is \emph{finite} CSP --- i.e.~the set of variables, the set of constraints and the domain of the variables are all finite. Although the computational cost of finite CSP is high (i.e.~the general problem is NP-complete), it can be solved in a finite time by a machine\footnote{For a general reference on solving classical CSP see \cite{rossi2006handbook}}. 

Unfortunately, when it comes to problems concerning behaviours of autonomous systems, the classical variant is too restrictive. Such problems can be naturally specified as CSP with infinite sets of variables (corresponding to the states of a system) and infinite sets of constraints (corresponding to the transitions between the states of a system). In recent years, we have witnessed a giant progress in solving infinite variants of CSP. The authors of \cite{bodirsky2006constraint} (see also a survey article \cite{bodirsky2008constraint}) applied algebraic and model-theoretic tools to analyze CSP over infinite domains. This research inspired the Warsaw Logical Group to investigate, so called, locally finite CSP --- i.e.~CSP over finite domain, whose variables and constraints form a definable infinite set of finite arities (see \cite{DBLP:conf/lics/KlinKOT15} and \cite{ochremiak2016extended}). They showed that that CSP defined in the first-order theory of rational numbers with their natural ordering can be solved effectively\footnote{To be more precise, they worked in the maximal tight extension of the theory, see \cite{licsMRP} for more details.}. The key technical observation was a property of definability in rational numbers, which we reformulate as the following axiom.
\begin{axiom}[DEF-CSP]
A definable CSP over a finite domain has a definable solution if and only if it has a solution. 
\end{axiom}
It was further observed in \cite{bodirsky2013decidability} that infinite CSP --- i.e.~CSP whose domains, variables and constraints form infinite definable sets reduce to locally finite CSP (an explicit reduction is given in Section~4 of \cite{DBLP:conf/fsttcs/KlinLOT16}). For this reason, without loss of generality, we can focus on CSP over finite domains. An example of such CSP is the problem of 3-colorability of graphs.

\begin{example}[3-colorability of an infinite graph]\label{e:3col}
Consider the following infinite graph definable over natural numbers $\struct{N}$ with equality:
\begin{eqnarray*}
V &=& \{\tuple{a, b} \in N^2 \colon a \neq b \} \\
E &=& \{\tuple{\tuple{a,b},\tuple{c,d}} \in V \times V \colon (a = d \wedge b \neq c) \vee (a \neq d \wedge b = c)  \}
\end{eqnarray*}
One may wander if this infinite graph is 3-colorable. Figure~\ref{f:3col} gives the negative answer by exhibiting a finite subgraph, which is not 3-colorable.  This problem fits into the framework of CSP as follows: the domain $D = \{Y, G, B\}$ consists of three constants $Y, G, B$, the set of variables $V$ is the set of vertices $V$ of the graph, and the set of constraints is given as $C = \{\tuple{\tuple{x, y}, {\neq_D}} \colon \tuple{x, y} \in E \}$, where ${\neq_D} \subset D \times D$ is the inequality relation on $D$. Observe, that every set definable in $\struct{N}$, can be treated as a set definable in the rational numbers with their natural ordering. Therefore, we can use the machinery of \cite{DBLP:conf/lics/KlinKOT15} to solve such problems effectively. 
\end{example}
\begin{figure}[htb]
    \centering
    \resizebox{0.4\textwidth}{!}{%
    \begin{tikzpicture}[ball/.style = {circle, draw, align=center, anchor=north, inner sep=0}]
\node[ball,text width=0.8cm,fill=blue!30] (ball12) at (2,0) {{$1,2$}};
\node[ball,text width=0.8cm,fill=green!30] (ball51) at (5,0) {{$5,1$}};
\node[ball,text width=0.8cm,fill=yellow!30] (ball25) at (3.5,-1) {{$2,5$}};
\node[ball,text width=0.8cm,fill=green!30] (ball31) at (2,-2) {{$3,1$}};
\node[ball,text width=0.8cm,fill=blue!30] (ball14) at (5,-2) {{$1,4$}};
\node[ball,text width=0.8cm,fill=yellow!30] (ball23) at (0.5,-3) {{$2,3$}};
\node[ball,text width=0.8cm,fill=yellow!30] (ball45) at (6.5,-3) {{$4,5$}};
\node[ball,text width=0.8cm] (ball34) at (3.5,-5) {$3,4$};

\draw[edgeArrow] (ball34) to (ball31);
\draw[edgeArrow] (ball34) to (ball14);

\node[ball,text width=0.8cm,fill=green!30] (ball42) at (2.5,-3.5) {{$4,2$}};
\node[ball,text width=0.8cm,fill=blue!30] (ball53) at (4.5,-3.5) {{$5,3$}};

\draw[edgeArrow] (ball34) to (ball42);
\draw[edgeArrow] (ball34) to (ball53);
\draw[edgeArrow] (ball34) to (ball23);
\draw[edgeArrow] (ball34) to (ball45);
\draw[edgeArrow] (ball12) to (ball51);
\draw[edgeArrow] (ball12) to (ball25);
\draw[edgeArrow] (ball51) to (ball25);
\draw[edgeArrow] (ball12) to (ball31);
\draw[edgeArrow] (ball51) to (ball14);
\draw[edgeArrow] (ball31) to (ball14);
\draw[edgeArrow] (ball25) to (ball42);
\draw[edgeArrow] (ball25) to (ball53);
\draw[edgeArrow] (ball12) to (ball23);
\draw[edgeArrow] (ball51) to (ball45);
\draw[edgeArrow] (ball14) to (ball42);
\draw[edgeArrow] (ball31) to (ball53);
\draw[edgeArrow] (ball23) to (ball42);
\draw[edgeArrow] (ball23) to (ball31);
\draw[edgeArrow] (ball45) to (ball14);
\draw[edgeArrow] (ball45) to (ball53);

\end{tikzpicture} 
    }%
    \caption{Counterexample to 3-colorability.}\label{f:3col}
\end{figure}
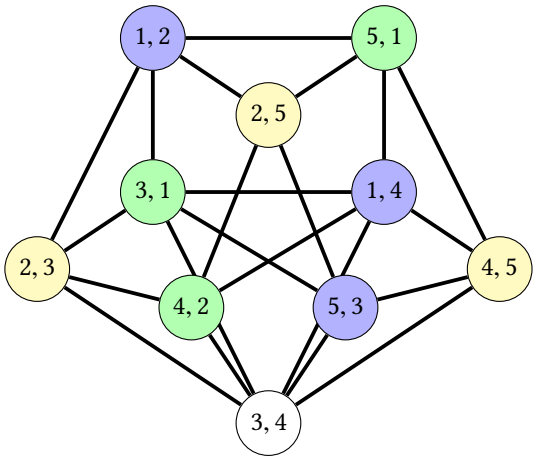

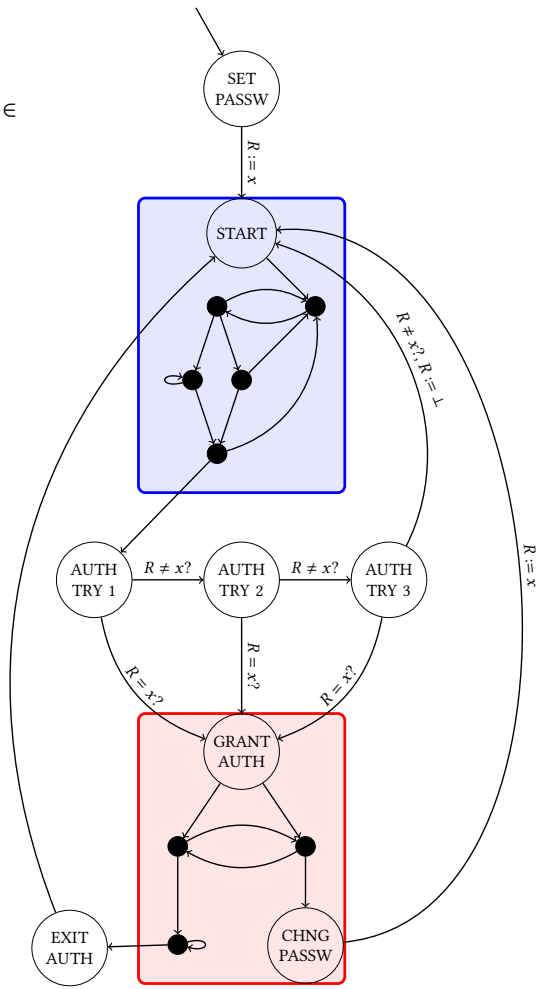
\begin{figure}[htb]
    \centering
    \resizebox{0.5\textwidth}{!}{%
 \begin{tikzpicture}[sloped,ball/.style = {circle, draw, align=center, anchor=north, inner sep=0}]
    
    \draw[blue,ultra thick,rounded corners,fill=blue!10] (0.9,-5.5) rectangle (5.1,-11.5);
    \draw[red,ultra thick,rounded corners,fill=red!10] (0.9,-16) rectangle (5.1,-21.5);
\node[text width=1.4cm] (init) at (2,-1.5) {{}};
\node[ball,text width=1.4cm] (passw) at (3,-2.5) {{SET\\PASSW}};
\node[ball,text width=1.4cm] (start) at (3,-5.5) {{START}};

\draw[thick,->]  (init) to  (passw) ;
\draw[thick,->]  (passw) to  node[midway,above] {$R := x$} (start);

\node[ball,text width=1.4cm] (try1) at (0,-12.5) {{AUTH\\TRY 1}};
\node[ball,text width=1.4cm] (try2) at (3,-12.5) {{AUTH\\TRY 2}};
\node[ball,text width=1.4cm] (try3) at (6,-12.5) {{AUTH\\TRY 3}};

\node[ball,text width=1.4cm] (access) at (3,-16) {{GRANT\\AUTH}};

\draw[thick,->]  (try1) to [bend right] node[midway,sloped,above] {$R = x$?} (access);
\draw[thick,->]  (try1) to  node[midway,sloped,above] {$R \neq x$?} (try2);
\draw[thick,->]  (try2) to  node[midway,sloped,above] {$R \neq x$?} (try3);
\draw[thick,->]  (try2) to  node[midway,sloped,above] {$R = x$?} (access);
\draw[thick,<-]  (access) to [bend right] node[midway,above] {$R = x$?} (try3);

\draw[thick,->]  (try3) to [bend right=50] node[midway,sloped,above] {$R \neq x?$, $R := \bot$} (start);

\node[ball,text width=1.4cm] (exit) at (-0.5,-20) {{EXIT\\AUTH}};
\draw[thick,->]  (exit) to [bend left=35] (start);

\node[ball,text width=0.4cm,fill=black] (a1) at (2.5,-7.5) {{}};
\node[ball,text width=0.4cm,fill=black] (a2) at (4.5,-7.5) {{}};
\node[ball,text width=0.4cm,fill=black] (a3) at (2,-9) {{}};
\node[ball,text width=0.4cm,fill=black] (a4) at (3,-9) {{}};
\node[ball,text width=0.4cm,fill=black] (a5) at (2.5,-10.5) {{}};

\draw[thick,->]  (a4) to (a5);

\draw[thick,->]  (start) to (a2);
\draw[thick,->]  (a1) to (a3);
\draw[thick,->]  (a1) to (a4);
\draw[thick,->]  (a3) to (a5);
\draw[thick,->]  (a4) to (a2);
\draw[thick,->]  (a5) [bend right=40] to (a2);
\draw[thick,->]  (a2) [bend left] to (a1);
\draw[thick,->]  (a1) [bend left] to (a2);


\draw[thick,->]  (a3) to [loop left] (a3);
\draw[thick,->]  (a5) to (try1);

\node[ball,text width=0.4cm,fill=black] (b1) at (1.7,-18.5) {{}};
\node[ball,text width=0.4cm,fill=black] (b2) at (4.3,-18.5) {{}};
\node[ball,text width=0.4cm,fill=black] (b3) at (1.7,-20.5) {{}};

\node[ball,text width=1.4cm] (change) at (4.3,-19.95) {{CHNG\\PASSW}};
\draw[thick,->]  (change) to [bend right=90] node[midway,sloped,above] {$R := x$} (start);

\draw[thick,->]  (access) to (b1);
\draw[thick,->]  (access) to (b2);
\draw[thick,->]  (b1) to (b3);
\draw[thick,->]  (b2) to (change);
\draw[thick,->]  (b3) to (exit);
\draw[thick,->]  (b2) [bend left] to (b1);
\draw[thick,->]  (b1) [bend left] to (b2);


\draw[thick,->]  (b3) to [loop right] (b3);

\end{tikzpicture} 
    }%
    \caption{A register machine that models access control to some parts of the system.}\label{f:register:machine}   

  \end{figure}
An important source of infinite graphs come from finite memory machines in the sense of Kaminski, Michael and Francez \cite{kaminski1994finite}.

\begin{example}[Access-control register machine]\label{e:register:machine}
Figure~\ref{f:register:machine} represents a register machine with one register $R$. The machine starts in state ``SET PASSW'', where it awaits for the user to provide a password $x$. This password is then stored in register $R$, and the machine enters state ``START''. Inside the blue rectangle the machine can perform actions that do not require authentication, whereas the actions that require authentication are presented inside the red rectangle. The red rectangle can be entered by the state ``GRANT AUTH'', which can be accessed from one of three authentication states. In order to authorise, the machine moves to state ``AUTH TRY 1'', where it gets input $x$ from the user. If the input is the same as the value previously stored in register $R$, then the machine enters state ``GRANT AUTH''. Otherwise, it moves to state ``AUTH TRY 2'' and repeats the procedure. Upon second unsuccessful authorisation, the machine moves to state ``AUTH TRY 3''. But  if the user provides a wrong password when the machine is in state ``AUTH TRY 3'', the register $R$ is erased (replaced with a value that is outside of the user's alphabet) --- preventing the machine to reach any of the states from the red rectangle.
Inside the red rectangle any action that requires authentication can be performed. For example, the user may request the change of the password.
Observe that in contrast to finite automata, the graph of possible configurations in a register machine is infinite. Nonetheless, we can formulate many properties of such graphs as CSP problems over natural numbers with equality and solve them effectively. 
\end{example}

\begin{remark}
Every constraint satisfaction problem can be presented as a pair of relational structures $\struct{V}$, $\struct{D}$ over the same relational signature $\Sigma$. This signature $\Sigma$ consists of a pair $\tuple{R, k}$ for every relation $R \subseteq D^k$ from a constraint $\tuple{\tuple{x_1, x_2, \dotsc, x_k}, R} \in C$. The interpretation of symbol $R/k \in \Sigma$ in $\struct{V}$ is $R^V(x_1, x_2, \dotsc, x_k) \Leftrightarrow \tuple{\tuple{x_1, x_2, \dotsc, x_k}, R} \in C$, and the interpretation in $\struct{D}$ is the relation $R$ itself. Moreover, a solution $\mor{S}{V}{D}$ to the CSP is a homomorphism from $\struct{V}$ to $\struct{D}$.
\end{remark}

To understand these results, we have to recall some basic concepts from model and set theory. We shall do this in Section~\ref{s:zfa}. The aim of this paper is to reverse the theorem stated as Axiom~DEF-CSP and give the full characterisation of set theories with atoms where Axiom~DEF-CSP holds. But we shall do much more: in Section~\ref{s:boolean:toposes} we reformulate Axiom~DEF-CSP as an intrinsic Axiom of any Boolean topos (Axiom~CSP) and show that it is equivalent to another, well-known, axiom: Boolean prime ideal theorem. Then in Subsection~\ref{ss:characterisation} we show that Axiom~CSP holds in $\cont{\aut{A \sqcup \word{A_0}}}$ for every finite $A_0 \subset A$ if and only if the automorphism group of $\struct{A}$ is extremely amenable. By the transfer principle we  conclude that this is equivalent to Axiom~DEF-CSP in $\catw{ZFA}(\struct{A})$. In Section~\ref{s:non:boolean:toposes} we investigate Axiom~CSP in non-Boolean toposes pointing out many obstacles to the equivalence between Boolean prime ideal theorem and possible formulations of Axiom~CSP.

\section{Set theory with atoms}
\label{s:zfa}

For any structure $\struct{A}$ we can build a von Neumann-like hierarchy of sets with elements from $\struct{A}$ \cite{mostowski1939unabhangigkeit}, \cite{halbeisen2017combinatorial}. The elements of $A$ will be thought of as ``atoms''.

\begin{definition}[The cumulative hierarchy of sets with atoms]
Let $\struct{A}$ be an algebraic structure with universe $A$. Consider the following sets defined by transfinite recursion:
\begin{itemize}
    \item $V_0(\struct{A}) = A$
    \item $V_{\alpha + 1}(\struct{A}) = \mathcal{P}(V_{\alpha}(\struct{A})) \cup V_{\alpha}(\struct{A})$
    \item $V_{\lambda}(\struct{A}) = \bigcup_{\alpha < \lambda} V_{\alpha}(\struct{A})$ if $\lambda$ is a limit ordinal
\end{itemize}
Then the cumulative hierarchy of sets with atoms $\struct{A}$ is defined as $V(\struct{A}) = \bigcup_{\alpha \colon \mathit{Ord}} V_{\alpha}(\struct{A})$.
\end{definition}

Observe, that the universe $V(\struct{A})$ carries a natural action $\mor{(\bullet)}{\aut{\struct{A}} \times V(\struct{A})}{V(\struct{A})}$ of the automorphism group $\aut{\struct{A}}$ of structure $\struct{A}$ --- it is just applied pointwise to the atoms of a set. If $X \in V(\struct{A})$ is a set with atoms then by its set-wise stabiliser we shall mean the set: $\aut{\struct{A}}_X = \{\pi \in \aut{\struct{A}} \colon \pi \bullet X = X \}$; and by its point-wise stabiliser the set: $\aut{\struct{A}}_{(X)} = \{\pi \in \aut{\struct{A}} \colon \forall_{x \in X} \pi \bullet x = x \}$. Moreover, for every $X$, these sets inherit a group structure from $\aut{\struct{A}}$.

There is an important sub-hierarchy of the cumulative hierarchy of sets with atoms $\struct{A}$, which consists of ``symmetric sets'' only. To define this hierarchy, we have to equip $\aut{\struct{A}}$ with a structure of a topological group.

\begin{definition}[Symmetric set]
 A set $X \in V(\struct{A})$ is \emph{symmetric} if the set-wise stabilisers of all of its descendants $Y$ is an open set (an open subgroup of $\aut{\struct{A}}$), i.e.~for every $Y \in^* X$ we have that: $\aut{\struct{A}}_{Y}$ is open in $\aut{\struct{A}}$, where ${\in^*}$ is the reflexive-transitive closure of the membership relation ${\in}$.
\end{definition}

Of a special interest is the topology on $\aut{\struct{A}}$ inherited from the product topology on $A^A$. We shall call this topology the canonical topology on $\aut{\struct{A}}$. In this topology, a subgroup $\group{H}$ of $\aut{\struct{A}}$ is open if there is a finite $A_0 \subseteq A$ such that: $\aut{\struct{A}}_{(A_0)} \subseteq \group{H}$, i.e.: group $\group{H}$ contains a pointwise stabiliser of some finite set of atoms.

\begin{definition}[Sets with atoms]
The sub-hierarchy of $V(\struct{A})$ that consists of symmetric sets according to the canonical topology on $\aut{\struct{A}}$ will be denoted by $\catw{ZFA}(\struct{A})$.
\end{definition}

Here are some standard examples of sets with atoms.

\begin{example}[The basic Fraenkel-Mostowski model]\label{e:first:zfa}
Let $\struct{N}$ be the structure of natural numbers with equality. We call $\catw{ZFA}(\struct{N})$ the basic Fraenkel-Mostowski model of set theory with atoms. Observe that $\aut{\struct{N}}$ is the group of all bijections (permutations) on $N$. 
The following are examples of sets in $\catw{ZFA}(\struct{N})$:
\begin{itemize}
\item all sets without atoms $N$, e.g.~$\emptyset, \{\emptyset\}, \{\emptyset, \{\emptyset\}, \dotsc\}, \dotsc$
\item all finite subsets of $N$, e.g.~$\{0\}, \{0,1,2,3\}, \dotsc$
\item all cofinite subsets of $N$, e.g.~$\{1, 2, 3, \dotsc\}, \{4, 5, 6, \dotsc\}, \dotsc$
\item $N\times N$
\item $\{\tuple{a,b} \in N^2 \colon a \neq b\}$
\item $N^* = \bigcup_{k\in N} N^k$ 
\item $\mathcal{K}(N) = \{N_0 \colon N_0 \subseteq N, \textit{$N_0$ is finite}\}$
\item $\mathcal{P}_s(N) = \{N_0 \colon N_0 \subseteq N, \textit{$N_0$ is symmetric}\}$
\end{itemize}
Here are examples of sets in $V(\struct{N})$ which are not symmetric:
\begin{itemize}
\item $\{0, 2, 4, 6, \dotsc\}$
\item $\{\tuple{n, m} \in N^2 \colon n \leq m\}$
\item the set of all functions from $N$ to $N$
\item $\mathcal{P}(N) = \{N_0 \colon N_0 \subseteq N\}$
\end{itemize}
\end{example}

\begin{example}[The ordered Fraenkel-Mostowski model]\label{e:ordered:zfa}
Let $\struct{Q}$ be the structure of rational numbers with their natural ordering. We call $\catw{ZFA}(\struct{Q})$ the ordered Fraenkel-Mostowski model of set theory with atoms. Observe that $\aut{\struct{Q}}$ is the group of all order-preserving bijections on $Q$.
All symmetric sets from Example~\ref{e:first:zfa} are symmetric sets in $\catw{ZFA}(\struct{Q})$ when $N$ is replaced by $Q$. Here are some further symmetric sets:
\begin{itemize}
\item $\{\tuple{p, q} \in Q^2 \colon p \leq q\}$
\item $\{\tuple{p, q} \in Q^2 \colon 0 \leq p \leq q \leq 1 \}$
\end{itemize}
\end{example}

Observe that the group $\aut{\struct{A}}_{(A_0)}$ is actually the group of automorphism of structure $\struct{A}$ extended with constants $A_0$, i.e.: $\aut{\struct{A}}_{(A_0)} = \aut{\struct{A} \sqcup A_0}$. Then a set $X \in V(\struct{A})$ is symmetric if and only if there is a finite $A_0 \in A$ such that $\aut{\struct{A} \sqcup A_0} \subseteq \aut{\struct{A}}_X$ and the canonical action of topological group $\aut{\struct{A} \sqcup A_0}$ on discrete set $X$ is continuous. A symmetric set is called $A_0$-equivariant (or equivariant in case $A_0 = \emptyset$) if $\aut{\struct{A} \sqcup A_0} \subseteq \aut{\struct{A}}_X$. Therefore, the (non-full) subcategory of $\catw{ZFA}(\struct{A})$ on $A_0$-equivariant sets and $A_0$-equivariant functions is equivalent to the category $\cont{\aut{\struct{A} \sqcup \word{A_0}}} \subseteq \catw{Set}^{\aut{\struct{A} \sqcup \word{A_0}}}$ of continuous actions of the topological group $\aut{\struct{A} \sqcup A_0}$ on discrete sets.

\begin{example}[Equivariant sets]
In the basic Fraenkel-Mostowski model:
\begin{itemize}
\item all sets without atoms are equivariant
\item all finite subsets $N_0 \subseteq N$ are $N_0$-equivariant
\item all finite subsets $N_0 \subseteq N$ are $(N \setminus N_0)$-equivariant
\item $N\times N, N^{(2)}, \mathcal{K}(N), \mathcal{P}_S(N)$ are equivariant
\end{itemize}
\end{example}

\begin{definition}[Definable set]
We shall say that an $A_0$-equivariant set $X \in \catw{ZFA}(\struct{A})$ is definable if its canonical action has only finitely many orbits, i.e.~if the relation $x \equiv y \Leftrightarrow \exists_{\pi \in \aut{\struct{A} \sqcup \word{A_0}}} \; x = \pi \bullet y$ has finitely many equivalence classes.
\end{definition}

For an open subgroup $\group{H}$ of $\aut{\struct{A}}$ let us denote by $\aut{\struct{A}}/\group{H}$ the quotient set $\{\pi \group{H} \colon \pi \in \aut{\struct{A}} \}$. This set carries a natural continuous action of $\aut{\struct{A}}$, i.e.~for $\sigma, \pi \in \aut{\struct{A}}$, we have $\sigma \bullet \pi \group{H} = (\sigma \circ \pi) \group{H}$. All transitive (i.e.~single orbit) actions of $\aut{\struct{A}}$ on discrete sets are essentialy of this form (see for example Chapter~III, Section~9 of \cite{maclane2012sheaves}). Therefore, equivariant definable sets are essentially finite unions of sets of the form $\aut{\struct{A}}/\group{H}$. Moreover, if structure $\struct{A}$ is $\omega$-categorical, then equivariant definable sets are the same as sets definable in the first order theory of $\struct{A}$ extended with elimination of imaginaries \cite{licsMRP}.

\begin{definition}[Ramsey property]
A structure $\struct{A}$ has a Ramsey property if for every open subgroup $\group{H}$ of $\aut{\struct{A}}$, every function $\mor{f}{\aut{\struct{A}}/\group{H}}{\{1, 2, \dotsc, k\}}$ and every finite set $C \subseteq \aut{\struct{A}}/\group{H}$ there is $\pi \in \aut{\struct{A}}$ such that $f$ is constant on $g\bullet C$, i.e.~there exists $0 \leq i \leq k$ such that for all $c \in C$ we have $f(\pi \bullet c) = i$. 
\end{definition}

The authors of \cite{DBLP:conf/lics/KlinKOT15} working in the ordered Fraenkel-Mostowski model $\catw{ZFA}(\struct{Q})$, showed that an equivarian definable constraint satisfaction problem has a solution if and only if it has an equivariant definable solution. A careful inspection of their proof shows that this result can be strengthen to all equivariant sets. The proof is based their results on a recently discovered result in topological dynamic \cite{pestov1998free}. We shall show that this advanced result is not needed at all. Before that, let us recall a very old problem about the independence of the Axiom of Choice from other axioms.

\begin{definition}[Ideal]
\label{d:ideal}
Let $\struct{B}$ be a Boolean algebra. An ideal in $\struct{B}$ is a proper subobject $I \subset B$ satisfying the following conditions:
\begin{itemize}
\item if $a, b \in I$ then $a \vee b \in I$
\item if $a \in I$ then for every $b \in B$ such that $b \leq a$ we have that $b \in I$ 
\end{itemize}
\end{definition}

\begin{definition}[Prime ideal]
\label{d:prime:ideal}
Let $I$ be an ideal in $\struct{B}$. We say that $I$ is prime if for every $b \in B$ either $b \in I$ or $\neg b \in I$.
\end{definition}

The Boolean Prime Ideal Theorem (BPIT) states that every ideal in Boolean algebra can be extended to a prime ideal. It is a routine to check that BPIT is follows from the Axiom of Choice \cite{jech2008axiom}, \cite{howard1998consequences}. It was a long-standing open problem whether the reverse implication holds as well. In 1964 Halpern \cite{halpern1964independence} used a model of ZFA over the rational numbers with the canonical ordering (nowadays called the ordered FM model $\catw{ZFA}(\struct{Q})$) to prove that the Axiom of Choice is not a consequence of BPIT in set theory with atoms. That is, he showd that in $\catw{ZFA}(\struct{Q})$ the Axiom of Choice fails badly, but BPIT holds. This result was later amplified in \cite{halpern1971boolean} to give the first proof that the Axiom of Choice is not a consequence of BPIT in ZF (without atoms).

\begin{remark}
An ideal $I$ in $\struct{B}$ can be represented by a homomorphism $\mor{h_I}{\struct{B}}{\struct{B}'}$ to a Boolean algebra $\struct{B}'$, i.e.~$I = h_I^{-1}(0)$. A prime ideal is an ideal that can be represented by a homomorphism to $2$ equipped with the usual Boolean algebra structure. Therefore, an ideal $I$ in $\struct{B}$ can be extended to a prime ideal $P$ iff $\struct{B}'$ has a prime ideal $J$. In this case, $h_P = h_J \circ h_I$. This means, that BPIT is equivalent to the statement that every non-trivial Boolean algebra has a prime ideal. We shall use this characterisation in Section~\ref{s:boolean:toposes}.
\end{remark}

It is the result of Halpern that we use to prove that in the cumulative hierarchy $V(\struct{Q})$ the following holds: ``an $X$-equivariant CSP has a solution if and only if it has an $X$-equivariant solution''. This may be formulated as ZFA-Axiom~CSP.
\begin{axiom}[ZFA-Axiom~CSP]
An $X$-equivariant CSP has an $X$-equivariant solution if and only if it has a solution. 
\end{axiom}
But, in fact, we do more. First, we reformulate ZFA-Axiom~CSP as an intrinsic property of a topos and call it Axiom~CSP. Then, we show that in Boolean toposes Axiom~CSP is actually equivalent to BPIT. In particular, for every set of atoms $\struct{A}$ we have that $\catw{ZFA}(\struct{A})$ satisfies Axiom~CSP if and only if it satisfies BPIT. This will give (1) and (2) from the abstract, with the one caveat: equivariance is not an intrinsic property of $\catw{ZFA}(\struct{A})$, therefore we have to state Axiom~CSP in every $\cont{\aut{\struct{A} \sqcup \word{A_0}}}$, and then by the transfer principle (see \cite{licsMRP}) recover the desired property.

A careful inspection of the proof of Halpern \cite{halpern1964independence} shows that the crucial property of $\struct{Q}$  is that  $\aut{\struct{Q}}$ has the Ramsey property.
This was further explored in \cite{johnstone1984almost} and in full details in \cite{blass1986prime}. Moreover, Theorem~2 of \cite{blass1986prime} states that the Ramsey property of $\aut{\struct{A}}$ is equivalent to BPIT in $\catw{ZFA}(\struct{A})$. Therefore, Ramsey property of $\aut{\struct{A}}$ is also equivalent to Axiom~CSP in $\catw{ZFA}(\struct{A})$. This is, however, not enough from the reason mentioned in the above: the literal translation of BPIT to $\catw{ZFA}(\struct{A})$ says that every symmetric ideal on a symmetric Boolean algebra can be extended to a symmetric prime ideal. This statement is weaker than: ``every $A_0$-equivariant ideal on an $A_0$-equivariant Boolean algebra can be extended to an $A_0$-equivariant prime ideal''. Fortunately, inspection of the proof \cite{blass1986prime} shows that the constructed prime ideal is, in fact, $A_0$-equivariant.

In 2005 Kechris, Pestov and Todorcevic in their famous work on topological dynamic \cite{kechris2005fraisse} showed that for countable single-sorted structures $\struct{A}$ the Ramsey property for $\struct{A}$ is equivalent to \emph{extreme amenability} of $\aut{\struct{A}}$.

\begin{definition}[Extremely amenable group]
\label{d:extremely}
A topological group $\group{G}$ is called extremely amenable if its every action $\mor{(\bullet)}{\group{G} \times X}{X}$ on a non-empty compact Hausdorff space $X$ has a fixed point.
\end{definition}

Therefore, for countable single-sorted structures $\struct{A}$ BPIT in $\catw{ZFA}(\struct{A})$ is equivalent to the extreme amenability of $\aut{\struct{A}}$. This was first observed by Andreas Blass in 2011 in \cite{blass2011partitions}. Furthermore, Proposition~4.7 in \cite{kechris2005fraisse} says that the class of such structures coincides with the class of structures that arise as the Fraisse limit of a Fraisse order class with the Ramsey property.

From the perspective of effective computation in set with atoms, structure $\struct{A}$ have to be countable, thus the restriction in the above equivalence to countable structures only is not severe. Moreover, the Fraisse limit (over a relational signature) is always $\omega$-categorical, a property crucial for the termination of certain while-programs (see \cite{licsMRP} for more discussion).

\section{The axiom in Boolean toposes}
\label{s:boolean:toposes}

In this section we shall work in the internal language of a Boolean topos. A reader who is not familiar with the notion of the internal language may read the proofs as taking place in ZFA minus the axiom of extensionality\footnote{A set in a non-well-pointed topos may have more content than mere elements} (e.g.~$\emptyset$-equivariant sets with atoms). We shall be extra careful when defining set-theoretic concepts, such as finiteness, or a prime ideal. Although, in Boolean toposes many different definitions of such concepts coincide, this would not be the case for non-Boolean toposes studied in the next section.

\begin{definition}[Kuratowski finiteness]\label{d:finiteness}
Let $A$ be a set. By $K(A)$ we shall mean the sub-join-semilatice of the powerset $P(A)$ generated by singletons and the empty set. A set $A$ is Kuratowski-finite if it is the top element in $K(A)$.
\end{definition}
For the rest of this section we shall just write finite set for Kuratowski-finite set. The chief idea behind the above definition is that since a non-empty finite set $A$ can be constructed from singletons by taking binary unions, we have a certain induction principle. Let us assume that: (base of the induction) $\phi$ holds for singletons, and (step of the induction) whenever $\phi$ holds for $A_0 \subseteq A$ and $A_1 \subseteq A$ then $\phi$ holds for $A_0 \cup A_1$, then (conclusion) $\phi$ holds for $A$. For example, we can show that the Axiom of Choice internally holds for finite sets. To see this, recall the usual reformulation of AC for finite sets: every surjection $\mor{e}{X}{Y}$ onto a finite set $Y$ has a section $\mor{s}{Y}{X}$, i.e.: $e \circ s = \id{Y}$. Let us assume that $\mor{e}{X}{Y}$ is a surjection. Then for every finite $D$, the function $\mor{e^D}{X^D}{Y^D}$, where $e^D(h) = e \circ h$, is also a surjection. This can be proven by induction over $D$. If $D$ is the empty set, or a singleton, then the claim clearly holds. 
Therefore, let us assume the claim holds for finite $D_0$, $D_1$ and show that it also holds for $D_0 \cup D_1$. Since the topos is Boolean, without the loss of generality, we may assume that $D_0$ and $D_1$ are disjoint. The function $\mor{e^{D_0 \cup D_1}}{X^{D_0 \cup D_1}}{Y^{D_0 \cup D_1}}$ decomposes on disjoint $\mor{e^{D_0}}{X^{D_0}}{Y^{D_0}}$ and $\mor{e^{D_1}}{X^{D_1}}{Y^{D_1}}$ with $e^{D_0 \cup D_1} = e^{D_0} \times e^{D_1}$. Because the Cartesian product of two surjections is a surjection, we may infer that $e^{D_0 \cup D_1}$ is a surjection, what completes the step of the induction. Therefore, if $\mor{e}{X}{Y}$ is a surjection then for every finite $D$ we have that $\mor{e^D}{X^D}{Y^D}$ is a surjection. By setting $D = Y$, we obtain that $\mor{e^Y}{X^Y}{Y^Y}$ is a surjection and so for every $i \in Y^Y$ there exists $h \in X^Y$ such that $e^Y(h) = i$. In particular, for $\id{Y} \in Y^Y$ there exists $s \in X^Y$ such that $e^Y(s) = \id{Y}$. But, $e^Y(s) = e \circ s$, what completes the proof.

\begin{definition}[Finitary relation]
For  sets $A, B$ we shall call $K(A \times B) \subseteq P(A \times B)$ the set of finitary relations from $A$ to $B$.  A finitary relation $R$ is a partial function if the following holds: $R(a, b) \wedge R(a, b') \vdash b = b'$. We shell denote the set of finitary partial functions from $A$ to $B$ by $B^{\underline{A}}$.
\end{definition}
In a Boolean topos a subset of a finite set is finite, therefore if $A$ and $B$ are finite, then a finitary relation from $A$ to $B$ is just a relation from $A$ to $B$.

\begin{definition}[Finitary homomorphism]
Let $\struct{A}$ and $\struct{B}$ be two relational structures over a common signature $\Sigma$. A relational homomorphism from $\struct{A}$ to $\struct{B}$ is a relation $f \in  P(A \times B)$ that preserves all relations $R/k \in \Sigma$, i.e.: $$f(a_1, b_1) \wedge f(a_2, b_2) \wedge \dotsc \wedge f(a_k, b_k) \wedge R(a_1, a_2, \cdots, a_k) \vdash R(b_1, b_2, \cdots, b_k)$$
A finitary homomorphism is a finitary partial function which is also a relational homomorphism. The set of all finitary  homomorphisms from $\struct{A}$ to $\struct{B}$ will be denoted by $\hom(\underline{\struct{A}}, \struct{B})$.
\end{definition}

Let us observe that there is a morphism $\mor{\gamma_0}{\hom(\underline{\struct{A}}, \struct{B})}{K(A)}$ that assigns to a finitary homomorphism $h \in \hom(\underline{\struct{V}}, \struct{D})$ its domain $\gamma_0(h) \in K(A) \subseteq P(A)$.

\begin{definition}[Jointly-total homomorphisms]
Let $\struct{A}$ and $\struct{B}$ be two relational structures over a common signature $\Sigma$. We shall say that a set of finitary homomorphisms $H \subseteq \hom(\underline{\struct{A}}, \struct{B})$ is jointly total if every finite $A_0 \in K(A)$ is  a subdomain of a finitary homomorphisms from $H$, i.e.: $\exists_{h \in H} A_0 \subseteq \gamma_0(h)$.
\end{definition}

Now, we are ready to state Axiom~CSP in Boolean toposes.

\begin{axiom}[CSP]
For every relational signature $\Sigma$ and a pair of structures $\struct{V}$ and $\struct{D}$ over $\Sigma$ such that $D$ is a finite cardinal, the following are equivalent:
\begin{itemize}
\item there exists a homomorphism from $\struct{V}$ to $\struct{D}$
\item the set of finitary homomorphisms $\hom(\underline{\struct{V}}, \struct{D})$ is jointly total
\end{itemize}
\end{axiom}

In the below we shall show that Axiom~CSP in Boolean toposes is equivalent to Boolean prime ideal theorem (BPIT). Let us recall the terminology first.

\begin{definition}[Boolean algebra]
\label{d:boolean:algebra}
An algebra $\struct{B}$ is a structure $\tuple{B, 1, {\wedge}, {\neg}}$, where $1$ is a constant, $\mor{{\wedge}}{B\times B}{B}$ is a binary operation, and $\mor{{\neg}}{B}{B}$ is an unary operation. Consider relation ${\leq} \subseteq B \times B$ defined as: $a \leq b \Leftrightarrow a = a \wedge b$. We say that $\struct{B}$ is a Boolean algebra if the following holds:
\begin{itemize}
\item ${\leq}$ is a partial order on $B$ with finite joins given by $\wedge$ and the greatest element $1$ 
\item for every $b \in B$ we have that: $\neg \neg b = b$
\end{itemize}
\end{definition}
If $\struct{B}$ is a Boolean algebra, then $1$ is its internal true value, and operation $\wedge$ is the internal conjunction. Other operations in a Boolean algebra can be defined in the usual way:
\begin{itemize}
\item $0 = \neg 1$ for the false value
\item $a \vee b = \neg(\neg a \wedge \neg b)$ for the internal disjunction
\item $a \oplus b = (a \land b) \lor (\neg a \land \neg b)$
\end{itemize}

\begin{axiom}[BPIT]
For every non-trivial Boolean algebra $\struct{B}$ there is a homomorphism $\struct{B} \rightarrow 2$ to the initial Boolean algebra $2 = 1 \sqcup 1$.
\end{axiom}

The constraint satisfaction problem is defined over \emph{relational} structures. Therefore, to fit into the framework of CSP we should treat a Boolean algebra $\struct{B}$ as if it was defined over a relational signature, with an unary predicate $\word{top}(x) \Leftrightarrow x = 1$, ternary  predicate $\word{and}(x,y,z) \Leftrightarrow x \land y = z$ and binary predicate $\word{not}(x, y) \Leftrightarrow  \neg x = y$. The axioms should express that there exists unique $x$ that satisfy $\word{top}$ and that $\word{and}$ and $\word{not}$ are functional relations.

\begin{theorem}[Axiom~CSP implies BPIT]\label{t:csp:to:bpit}
Axiom~CSP implies BPIT in Boolean toposes.
\end{theorem}
\begin{proof}
Let $\struct{B}$ be a Boolean algebra. 
%
By Axiom~CSP, it suffices to show that the set of finitary homomorphisms $\hom(\underline{\struct{B}}, 2)$ is jointly total, i.e.~for every finite $B_0$ in $K(B)$ there exists a partial homomorphism $B_0 \subseteq B_1 \rightarrow 2$. We can assume that $B_1$ is closed under Boolean-algebra operations and still finite. The reason for that is that if $B_0$ is finite then in a Boolean topos $P(B_0)$ is finite as well (it coincides with $K(B_0)$). Because we have shown that the AC holds for finite sets, the standard proof of Zorn's Lemma can be carried over to our setting to show that $B_1$ has a maximal ideal, therefore (using again Boolean logic of the topos) $B_1$ has a prime ideal.
\end{proof}

\begin{theorem}[BPIT implies Axiom~CSP]\label{t:bpit:to:csp}
Axiom~BPIT implies Axiom~CSP in Boolean toposes.
\end{theorem}
\begin{proof}
Let us assume that $\struct{V}$ and $\struct{D}$ are structures over relational signature $\Sigma$. Furthermore, assume that $\struct{D} = \{0, 1, \dotsc, D-1\}$ is a finite cardinal and the set of finitary homomorphisms $\hom(\underline{\struct{V}}, \struct{D})$ is jointly total. We shall treat $\word{Var} = V \times D$ as a set of propositional variables. Consider the following subsets of propositions $F(\word{Var})$, where $F(\word{Var})$ is treated as the free Boolean algebra on $\word{Var}$:
\begin{itemize}
\item $T = \{ \tuple{v, 0} \lor \tuple{v, 1} \lor \cdots \lor \tuple{v, D-1} \colon v \in V \}$
\item $S = \{\neg (\tuple{v, n} \land \tuple{v, m}) \colon v \in V, n \in D, m \in D, n \neq m\}$
\item $C = \{\neg (\tuple{x_1, d_1} \land  \tuple{x_2, d_2} \land \cdots \land \tuple{x_n, d_n}) \colon R^V(x_1, x_2, \dotsc, x_n) \land \neg R^D(d_1, d_2, \dotsc, d_n)\}$
\end{itemize}
Consider the following set of propositions: $P = T \cup S \cup C$. Let us say that two propositions $\phi, \psi$ from $F(\word{Var})$ are equivalent if there is a finite $P_0 \subset P$ such that every valuation $V_0 \times D \rightarrow 2$ satisfying $P_0$ satisfies $\phi \oplus \psi$. Then $F(\word{Var})$ divided by this equivalence relation is again a Boolean algebra $F(\word{Var})/\equiv$ with the usual operations. We want to show that $F(\word{Var})/\equiv$ is non-trivial, i.e. $0 \neq 1$. Because every finite subset $P_0 \in P$ involves only finitely many variables $\word{Var}_0 \subseteq \word{Var}$, the set $V_0 = \gamma_0(\word{Var}_0) \subseteq V$ is finite. In fact, $P_0$ can be rewritten as the union of:
\begin{itemize}
\item $T_0 = \{ \tuple{v, 0} \lor \tuple{v, 1} \lor \cdots \lor \tuple{v, D-1} \colon v \in V_0 \}$
\item $S_0 = \{\neg (\tuple{v, n} \land \tuple{v, m}) \colon v \in V_0, n \in D, m \in D, n \neq m\}$
\item $C_0 = \{\neg (\tuple{x_1, d_1} \land  \tuple{x_2, d_2} \land \cdots \land \tuple{x_n, d_n}) \colon x_i \in V_0, R^V(x_1, x_2, \dotsc, x_n) \land \neg R^D(d_1, d_2, \dotsc, d_n)\}$
\end{itemize}
Since $V_0$ is finite, by the assumption, there exists a finitary homomorphism $h_0 \in \hom(\underline{\struct{V}}, \struct{D})$ with $V_0 \subseteq \gamma_0(h_0)$, which induces a valuation $V_0 \times D \rightarrow 2$. By the definition of the constraints, this valuation makes $P_0$ satisfiable. Therefore, every finite $P_0$ is satisfiable, and so $F(\word{Var})/\equiv$ is non-trivial. By Axiom~BPIT, there is a prime ideal $\mor{u}{F(\word{Var})/\equiv}{2}$, which composed with the canonical embedding $\mor{j}{F(\word{Var})}{F(\word{Var})/\equiv}$ gives a prime ideal $\overline{h} = u \circ j$ on $F(\word{Var})$. By the definition $\overline{h}$ maps propositions from $P$ to $1$. Consider the restriction $\mor{h}{V \times D}{2}$ of $\overline{h}$ to variables $\word{Var} = V \times D$. By propositions $T$ valuation $h$ is total and by propositions $S$ it is single-valued. Moreover, by propositions $C$ the valuation dos not violate any constraints. Therefore, $h$ is a homomorphism from $\struct{V}$ to $\struct{D}$. 
\end{proof}

\subsection{Characterisation theorems}
\label{ss:characterisation}

This subsection states our main characterisation theorems. We shall begin with a simple purely-logical proof that ZFA-Axiom~CSP holds in the ordered Fraenkel-Mostowski model of set theory with atoms. Observe, that we rely only on the old combinatorial result of Halpern \cite{halpern1964independence}.

\begin{theorem}[ZFA-Axiom~CSP in $\catw{ZFA}(\struct{Q})$]
ZFA-Axiom~CSP holds in $\catw{ZFA}(\struct{Q})$.
\end{theorem}
\begin{proof}
A careful inspection of the proof of Halpern \cite{halpern1964independence} shows that if a Boolean algebra $\struct{B}$ in $\catw{ZFA}(\struct{Q})$ is $X$-equivariant than it has an $X$-equivariant prime ideal. Therefore, BPIT holds in $\cont{\aut{\struct{Q} \sqcup \word{X}}}$. By Theorem~\ref{t:bpit:to:csp}, Axiom~CSP holds in $\cont{\aut{\struct{Q} \sqcup \word{X}}}$. Therefore, ZFA-Axiom~CSP holds in $\catw{ZFA}(\struct{Q})$.
\end{proof}

The proof of the next theorem is similar.
\begin{theorem}[ZFA-Axiom~CSP in $\catw{ZFA}(\struct{A})$]\label{t:zfa:axiom}
Let $\struct{A}$ be a countable structure. ZFA-Axiom~CSP holds in $\catw{ZFA}(\struct{A})$ if and only if the automorphism group of $\struct{A}$ is extremely amenable.  
\end{theorem}
\begin{proof}
A careful inspection of Theorem~2 of \cite{blass1986prime} states that the Ramsey property of $\struct{A}$ is equivalent to the property that every $X$-equivariant Boolean algebra $\struct{B}$ in $\catw{ZFA}(\struct{A})$ has an $X$-equivariant prime ideal.
Therefore, Ramsey property of $\struct{A}$ is also equivalent to ZFA-Axiom~CSP in $\catw{ZFA}(\struct{A})$. And by the result of Kechris, Pestov and Todorcevic \cite{kechris2005fraisse}, ZFA-Axiom~CSP in $\catw{ZFA}(\struct{A})$ is equivalent to extreme amenability of the automorphism group of $\struct{A}$.
\end{proof}

We shall now consider a weaker versions of Axiom~CSP and show that in continuous sets over a localic group it is equivalent to Axiom~CSP. 

\begin{definition}[Compact object]\label{d:compact}
An object $X$ of a cocomplete category $\catl{C}$ is called compact if its co-representation ${\mor{\hom_{\catl{C}}(X, -)}{\catl{C}}{\catw{Set}}}$ preserves filtered colimits of monomorphisms.
\end{definition}

\begin{axiom}[Compact~CSP]
For every relational signature $\Sigma$ and a pair of structures $\struct{V}$ and $\struct{D}$ over $\Sigma$ such that $D$ is a finite cardinal and $\struct{V}$ is compact, the following are equivalent:
\begin{itemize}
\item there exists a homomorphism from $\struct{V}$ to $\struct{D}$
\item the set of finitary homomorphisms $\hom(\underline{\struct{V}}, \struct{D})$ is jointly total
\end{itemize}
\end{axiom}

An object $A$ in $\cont{G}$ can be regarded as a relational structure $\tuple{A, (T_g)_{g \in \group{G}}}$, where $T_g(a, b) \Leftrightarrow g \circ a = b$ is a binary relation on $A$. Then a function $\mor{f}{A}{B}$ is a homomorphism $\mor{f}{\tuple{A, (T_g)_{g \in \group{G}}}}{\tuple{B, (T_g)_{g \in \group{G}}}}$ iff it is equivariant.  

\begin{theorem}[Compact~CSP implies Axiom~CSP in continuous sets]\label{t:compactcsp:csp}
Let $\group{G}$ be a localic group and $\cont{G}$ be the topos of its continuous actions on $\catw{Set}$. Then Compact~CSP holds in $\cont{G}$ iff Axiom~CSP holds in $\cont{G}$.
\end{theorem}
\begin{proof}
An object $A$ in $\cont{G}$ can be represented as a disjoint union of its orbits $A_x \in A/\group{G}$. By the definition of compactness for every finite set of orbits $F \subseteq  A/\group{G}$ the object $\bigcup F$ is compact.

Let us assume that every finite subset of $A$ has a solution, by  Axiom~CSP for compact objects every $\bigcup F$ has an equivariant solution. Moreover, if $F' \subset F$ then every equivariant solution of $F$ can be restricted to an equivariant solution of $F'$. Observe that $A \in \cont{G}$ can be regarded as a classical structure over a signature extended by relations $T_g(a, b) \Leftrightarrow a \circ g = b$. Then by Axiom~CSP for $\catw{Set}$, the classical structure over this extended signature has a solution. But, by definition of $T_g$, such solution must be equivariant.   
\end{proof}


We can summarize the above characterisations in the next theorem.

\begin{theorem}[Characterisation theorem]
Let $\struct{A}$ be a countable structure. Then the following are equivalent:
\begin{enumerate}
\item $\struct{A}$ is the Fraisse limit of a Fraisse order class with the Ramsey property
\item $\aut{\struct{A}}$ is extremely amenable
\item ZFA-Axiom~CSP holds in $\catw{ZFA}(\struct{A})$
\item Axiom~DEF-CSP holds in $\catw{ZFA}(\struct{A})$
\item Axiom~CSP holds in $\catw{ZFA}(\struct{A})$
\item Axiom~CSP holds in $\cont{\aut{A \sqcup \word{A_0}}}$ for every finite $A_0 \subset A$
\end{enumerate}
\end{theorem}
\begin{proof}
$(1) \Leftrightarrow (2)$ is Proposition~4.7 in \cite{kechris2005fraisse}. $(2) \Leftrightarrow (3)$ is the subject of Theorem~\ref{t:zfa:axiom}. $(3) \Leftrightarrow (4)$ is the consequence of Theorem~\ref{t:compactcsp:csp}. $(3) \Leftarrow (5)$ is trivial, and $(2) \Rightarrow (5)$ follows from Theorem~\ref{t:csp:to:bpit} and the main theorem of \cite{blass2011partitions}. $(3) \Leftrightarrow (6)$ is trivial.
\end{proof}

\section{The case of non-Boolean toposes}
\label{s:non:boolean:toposes}

When we move to non-Boolean toposes, we have to be extra careful when stating classical definitions and axioms, because in constructive mathematics classically equivalent statements may be far different. Fortunately for us, the concept of Boolean algebra, ideal and prime ideal move smoothly to the intuitionistic setting with one caveat: not every maximal ideal in a Boolean algebra has to be prime.

On the other hand, Axiom~CSP is much more difficult to handle in the intuitionistic setting. Actually, we have several different variants of Axiom~CSP depending on our interpretation of ``finiteness'' and admissible relational structures. Therefore, we should not expect that Axiom~CSP is equivalent to BPIT in constructive mathematics, because BPIT does not involve any notion of finiteness and there is not much concern about admissibility of Boolean algebra operations (however, we could take this into account). In general, the stronger the notion of ``finiteness'' and ``admissibility'' is, the stronger Axiom~CSP we obtain.

Let us discuss some possible definitions for an admissible structure $\struct{A}$:
\begin{enumerate}
\item Only complemented relations $R$ are admissible. That is, subobjects $\mor{s}{R}{A^k}$ such that there exists a subobject $\mor{\neg s}{R}{A^k}$ with the property that $s \cup \neg s = \id{A^k}$ and $s \cup \neg s = 0$.
\item $A$ is decidable. That is, the sobobject $\mor{\Delta}{A}{A \times A}$ that correspond to the equality predicate is complemented. Because, we assume that equality is always presented in the signature, decidability of $A$ is subsumed by the previous point.
\item All relations are admissible.  
\end{enumerate}

In the next subsections we discuss Axiom~CSP with respect to two internal notions of ``finiteness'': Kuratowski finiteness from Definition~\ref{d:finiteness} and Kuratowski subfiniteness (i.e.~being a subobject of a Kuratowski finite object).


\subsection{Kuratowski finiteness is too strong}

Consider Sierpienski topos $\catw{Set}^{\bullet \rightarrow \bullet}$. It is a routine to check that BPIT holds in $\catw{Set}^{\bullet \rightarrow \bullet}$, but Axiom~CSP does not hold even in case ``finiteness'' is interpreted as Kuratowski finiteness and only complemented relations $R$ are admissible. For a counterexample consider structures from Figure~\ref{f:csp:sierpienski}. The structure on the right side is the terminal object $1$ equipped with the empty unary relation $0 \rightarrow 1$. The structure on the left side is the only non-trivial subobject $\frac{1}{2}$ of $1$ equipped with the full unary relation $\frac{1}{2} \overset{\id{}}\rightarrow \frac{1}{2}$. There is a unique morphism $!$ from $\frac{1}{2}$ to $1$, but it is not a homomorphism, since it does not preserve the unary relation, i.e.~$! \circ \id{\frac{1}{2}} = \frac{1}{2} \neq 0$. On the other hand, the only Kuratowski finite subobject of $\frac{1}{2}$ is $0$ and the object of homomorphisms $\hom(0, 1)$ is isomorphic to $1$.

This example shows that Axiom~CSP with Kuratowski finiteness is too strong to be provable from BPIT, and too strong in general. 
If we weaken Axiom~CSP by weakening the notion of finiteness to Kuratowski subfiniteness then Axiom~CSP will hold even if all relations are admissible. The reason is that a structure $\mor{\struct{X}}{\struct{A}}{\struct{B}}$ in $\catw{Set}^{\bullet \rightarrow \bullet}$ can be encoded as a structure $\struct{A} \sqcup \struct{B}$ in $\catw{Set}$ with one additional relation encoding the graph of function $X$, i.e.~$R(x, y) \Leftrightarrow X(x) = y$ and one unary relation to distinguish domain from the codomain, i.e.~$S(x, y) = A \times A \sqcup B\times B$. Then for every finite substructure $A_0 \sqcup B_0$ of $\struct{A} \sqcup \struct{B}$ there is a finite substructure $A_0 \sqcup B_0 \subset A_1 \sqcup B_1 \subseteq \struct{A} \sqcup \struct{B}$ corresponding to a Kuratowski subfinite substructure of $\struct{X}$. Therefore, Axiom~CSP holds in $\catw{Set}^{\bullet \rightarrow \bullet}$ by the Axiom~CSP in $\catw{Set}$.

\subsection{Kuratowski subfiniteness is too weak}

Consider topos $\catw{Set}^{\bullet \leftarrow \bullet \rightarrow \bullet}$. Figure~\ref{f:csp:nonextreme} shows an example of a Boolean algebra, which does not have a prime ideal. Moreover, this example explicitly shows why we cannot carry over our proof of Theorem~\ref{t:csp:to:bpit} to constructive mathematics --- the Boolean algebra under consideration is Kuratowski finite, what means that in $\catw{Set}^{\bullet \leftarrow \bullet \rightarrow \bullet}$ not every finite Boolean algebra has a prime ideal. On the other hand, Axiom~CSP with Kuratowski finite subobjects fails and with Kuratowski subfinite subobjects holds for the same reasons as in the Sierpienski topos $\catw{Set}^{\bullet \rightarrow \bullet}$. 
Therefore, example from Figure~\ref{f:csp:nonextreme} shows that Axiom~CSP with Kuratowski subfiniteness is too weak to prove BPIT.

\begin{figure}[tb]
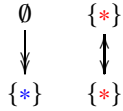


$$\bfig
\node a1(0, 0)[\{\textcolor{blue}{*}\}]
\node b1(300, 0)[\{\textcolor{red}{*}\}]
\node a2(0, 300)[\emptyset]
\node b2(300, 300)[\{\textcolor{red}{*}\}]

\arrow|r|/->>/[a2`a1;]
\arrow|r|/<->/[b2`b1;]

\efig$$
\caption{Axiom~CSP fails in $\catw{Set}^{\bullet \rightarrow \bullet}$ for Kuratowski finiteness. The structures are equipped with a single unary relation that holds on blue elements only.}
\label{f:csp:sierpienski}

\end{figure}

\begin{figure}[tb]
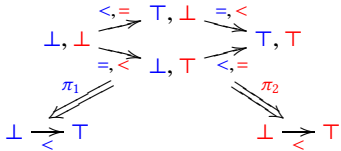

$$\bfig
\node bb(200, 400)[\textcolor{blue}{\bot}, \textcolor{red}{\bot}]
\node bt(600, 300)[\textcolor{blue}{\bot}, \textcolor{red}{\top}]
\node tb(600, 500)[\textcolor{blue}{\top}, \textcolor{red}{\bot}]
\node tt(1000, 400)[\textcolor{blue}{\top}, \textcolor{red}{\top}]

\node b1(0, 50)[\textcolor{blue}{\bot}]
\node t1(250, 50)[\textcolor{blue}{\top}]
\arrow|l|/@{->}@[blue]/[b1`t1;\textcolor{blue}{<}]

\node b2(950, 50)[\textcolor{red}{\bot}]
\node t2(1200, 50)[\textcolor{red}{\top}]
\arrow|l|/@{->}@[red]/[b2`t2;\textcolor{red}{<}]

\node f1(125, 100)[]
\node f2(400, 250)[]
\arrow|l|/@{=>}@[blue]/[f2`f1;\textcolor{blue}{\pi_1}]

\node f3(1025, 100)[]
\node f4(800, 250)[]
\arrow|r|/@{=>}@[red]/[f4`f3;\textcolor{red}{\pi_2}]

\arrow|l|/->/[bb`tb;\textcolor{blue}{<}, \textcolor{red}{=}]
\arrow|l|/->/[bb`bt;\textcolor{blue}{=}, \textcolor{red}{<}]
\arrow|r|/->/[bt`tt;\textcolor{blue}{<}, \textcolor{red}{=}]
\arrow|r|/->/[tb`tt;\textcolor{blue}{=}, \textcolor{red}{<}]
\efig$$
\caption{An example of a non-trivial finite Boolean algebra in $\catw{Set}^{\bullet \leftarrow \bullet \rightarrow \bullet}$ that has no prime ideal.}
\label{f:csp:nonextreme}
\end{figure}

\section{Conclusions and further work}
\label{s:conclusions}
In this paper we have given a simple purely-logical proof of ZFA-Axiom~CSP in the ordered Fraenkel-Mostowski model without using any advanced results from topology and model theory. Moreover, we have introduces an intrinsic characterisation of the statement ``definable CSP has a solution iff it has a definable solution'' and investigate it in general toposes. It turns out that in Boolean toposes this axiom is equivalent to Boolean prime ideal theorem, whereas in intuitionistic toposes there is no such an equivalence, nor an implication in either of the directions. It is an interesting question which positive-existential theories have classifying toposes validating Axiom~CSP; or more generally, in which Grothendieck toposes does Axiom~CSP hold. Finally, we reversed the main result of \cite{DBLP:conf/lics/KlinKOT15} by showing that Axiom DEF-CSP holds in  $\catw{ZFA}(\struct{A})$ if and only if $\aut{\struct{A}}$ is extremely amenable.    

%
\begin{acks} 
This research was supported by the National Science Centre, Poland, under projects 2018/28/C/ST6/00417.
\end{acks}

\bibliography{bibl}


\end{document}